\documentclass[letterpaper,12pt]{article}

\usepackage[left=1.3in, right=1.3in]{geometry}
\usepackage{amsmath,amssymb,amstext} % Lots of math symbols and environments
\usepackage[pdftex]{graphicx} % For including graphics N.B. pdftex graphics driver 
\usepackage{mathtools}
\usepackage{amsthm}
\usepackage{tikz}

\newtheorem{theorem}{Theorem}

\DeclarePairedDelimiter{\abs}{\lvert}{\rvert}

\renewcommand{\(}{\left(}
\renewcommand{\)}{\right)}

\usepackage{eepic}

\usepackage[pdftex,pagebackref=false]{hyperref}

\def\X{\mathcal{X}}
\def\Y{\mathcal{Y}}
\def\Z{\mathcal{Z}}
\def\W{\mathcal{W}}

\newcommand{\setft}[1]{\mathrm{#1}}

\newcommand{\density}[1]{\setft{D}\left(#1\right)}

\newcommand{\pos}[1]{\setft{Pos}\left(#1\right)}

\definecolor{White}{rgb}{1,1,1}
\definecolor{Black}{rgb}{0,0,0}
\definecolor{LightGray}{rgb}{.81,.81,.81}
\colorlet{ChannelColor}{LightGray}
\colorlet{ChannelTextColor}{Black}
\colorlet{ReadoutColor}{White}
\usetikzlibrary{calc}

\begin{document}

\title{Parallel repetition with a threshold in quantum interactive proofs}

\author{
  {\hspace{-0.2cm} Abel Molina}\\[2mm]
  {\it Institute for Quantum Computing and School of Computer Science}\\
  {\it University of Waterloo}\\
 }

\maketitle

\begin{abstract}
In this note, we show that $O(\log (1/\epsilon))$ rounds of parallel repetition with a threshold suffice to reduce completeness and soundness error to $\epsilon$ for single-prover quantum interactive proof systems. This improves on a previous $O(\log (1/\epsilon) \log \log (1/\epsilon))$ bound from Hornby (2018), while also simplifying its proof. A key element in our proof is a concentration bound from Impagliazzo and Kabanets (2010).
\end{abstract}

\section{Introduction}

When we have a probabilistic procedure that computes a binary output, one can often reduce the probability that the procedure errs by conducting parallel repetition with a threshold. The number of needed parallel instances scales as  $O( \log(1/\epsilon))$, with $\epsilon$ being the desired error. For example, this applies to those computations corresponding to the complexity classes \textbf{BPP}, \textbf{BQP}, and \textbf{IP}~\cite{arora2009computational}. Similar ideas can be made to work in the context of quantum interactive proof systems with perfect completeness \cite{kitaev2000parallelization}.

In the classical prover-verifier interactions corresponding to the class \textbf{IP}, one can use standard averaging arguments to show that it is optimal for a prover (who we shall call Bob) to play both deterministically and independently when his goal is to win at least $k$ out of $n$ parallel instances. By applying Chernoff bounds, this gives us a naive proof  of correctness for error reduction via parallel repetition with a threshold. 

In the case of quantum prover-verifier interactions, it is however not optimal in general for Bob to play independently when his goal is to win at least $k$ out of $n$ parallel instances, although it is optimal for the case where $k=n$  \cite{gutoski2010quantum}. In particular, a counterexample when $k=1$ and $n=2$  was provided in $\cite{molina2012hedging}$, with further examples in $\cite{arunachalam2018quantum, ganz2018quantum}$.

This means that in the context of computational complexity classes corresponding to quantum prover-verifier protocols with a single prover, there is no naive proof of correctness for error reduction via parallel repetition with a threshold. In particular, it means that it is not possible to trivially use Chernoff bounds in order to prove a decrease in the soundness error. 

For those quantum prover-verifier interactions with at least $3$ messages, one way to sidestep this complication as far as error reduction is concerned is to find an equivalent protocol with fewer or the same messages and perfect completeness \cite{kitaev2000parallelization}, and then perform parallel repetition where the prover is required to win all instances. For $2$ messages, one can make usage of a more complex error-reduction procedure \cite{jain2009two}, which considers several parallel instances and then divides them into subgroups for the purposes of determining a final outcome.

However, none of those results directly address the question of whether parallel repetition with a threshold is effective for the purpose of error reduction in quantum prover-verifier interactions with a single prover. This question was then addressed in \cite{hornby2018concentration}, where it was proved that for all such quantum prover-verifier protocols, parallel repetition with a threshold does, in fact, bring both soundness and completeness errors down to $0$ asymptotically. In particular, it is proved there that  $O(\log (1/\epsilon) \log \log (1/\epsilon)) $ parallel instances suffice in order to bring the completeness and soundness errors below $\epsilon$. This is proved by a reduction to the technique in \cite{jain2009two}.

We now prove here that $O(\log (1/\epsilon)$ number of parallel instances will suffice to bring the completeness and soundness errors below $\epsilon$. Furthermore, we are able to do this without an explicit reduction to any other setting. Our key technique will be the usage of a concentration bound of Impagliazo and Kabanets \cite{impagliazzo2010constructive}. The conditions needed for the application of the lemma correspond to the result in \cite{gutoski2010quantum} that looks at the situation where Bob aims to succeed in all parallel instances. This concentration bound has been previously used \cite{gavinsky2012quantum, pastawski2012unforgeable} to study attacks on certain prover-verifier interactions corresponding to quantum money protocols, while making usage of properties of the specific protocols rather than of the general result in~\cite{gutoski2010quantum}.  Note as well that the bound in \cite{impagliazzo2010constructive} can be seen as a rephrasing of a previous one in \cite{panconesi1997randomized}.

\section{Setting}

\label{sec:setting}

In the protocols that we study, verifier Alice and prover Bob exchange a series of messages, with Alice making a final binary decision. We assume without loss of generality that the first message is sent by Alice. Alice's actions are fixed by the protocol, while Bob is free to choose his actions. Formally, the protocol is represented by the following objects:

\begin{enumerate}

\item Two series of $r$ finite-dimensional Hilbert spaces $\X_1, \ldots, \X_r$ and $\Y_1, \ldots, \Y_r$, corresponding to the state spaces for the $r$ messages from Alice to Bob, and viceversa. 

\item  A series of $r$ finite-dimensional Hilbert spaces  $\Z_1, \ldots, \Z_r$, corresponding to Alice's internal memory after she sends each of her $r$ messages.

\item A density operator $\rho \in \density{\X_1 \otimes \Z_1}$, corresponding to Alice's first message to Bob.

\item A series of $r-1$ channels, which produce Alice's responses after each of the first $r-1$ messages from Bob. We label these channels as $\Psi_2, \ldots, \Psi_r$, with $\Psi_i$ mapping inputs in $\density{\Y_{i-1} \otimes \Z_{i-1}}$ to outputs in $\density{\X_{i}}$.

\item A binary POVM $\{P_0, P_1\} \subset \pos{\Y_r \otimes \Z_r}$, which determines the final outcome of the interaction. $0$ is identified as the \textit{losing} outcome for Bob, while $1$ is identified as the \textit{winning} outcome.
\end{enumerate}

\noindent A corresponding circuit diagram for the case where $r=3$ can be seen in Figure~\ref{fig:interaction}.

\begin{figure}
  \begin{center} \small
    \begin{tikzpicture}[scale=0.35, 
        turn/.style={draw, minimum height=14mm, minimum width=14mm,
          fill = ChannelColor, text=ChannelTextColor},
        emptyturn/.style={draw, densely dotted, minimum height=14mm, 
          minimum width=14mm, fill=black!4},
        measure/.style={draw, minimum width=7mm, minimum height=7mm,
          fill = ChannelColor},
        >=latex]
                
      \node (V0) at (-18,4) [turn] {$\rho$};      
      \node (V1) at (-8,4) [turn] {$\Psi_2$};
      \node (V2) at (2,4) [turn] {$\Psi_3$};
      \node (V3) at (12,4) [turn] {$\{P_0, P_1\}$};
      \node (M) at (16, 4)  {};

      \node (P0) at (-13,-1) [emptyturn] {$\Phi_1$};
      \node (P1) at (-3,-1) [emptyturn] {$\Phi_2$};
      \node (P2) at (7,-1) [emptyturn] {$\Phi_3$};
      
      \draw[->] ([yshift=4mm]V0.east) -- ([yshift=4mm]V1.west)
      node [above, midway] {$\Z_1$};
      
      \draw[->] ([yshift=4mm]V1.east) -- ([yshift=4mm]V2.west)
      node [above, midway] {$\Z_2$};
      
      \draw[->] ([yshift=4mm]V2.east) -- ([yshift=4mm]V3.west)
      node [above, midway] {$\Z_3$};

      \draw[->] ([yshift=-4mm]V0.east) .. controls +(right:20mm) and 
      +(left:20mm) .. ([yshift=4mm]P0.west) node [right, pos=0.4] {$\X_1$};
      
      \draw[->] ([yshift=-4mm]V1.east) .. controls +(right:20mm) and 
      +(left:20mm) .. ([yshift=4mm]P1.west) node [right, pos=0.4] {$\X_2$};
      
      \draw[->] ([yshift=-4mm]V2.east) .. controls +(right:20mm) and 
      +(left:20mm) .. ([yshift=4mm]P2.west) node [right, pos=0.4] {$\X_3$};
      
      \draw[->] ([yshift=4mm]P0.east) .. controls +(right:20mm) and 
      +(left:20mm) .. ([yshift=-4mm]V1.west) node [left, pos=0.6] {$\Y_1$};
      
      \draw[->] ([yshift=4mm]P1.east) .. controls +(right:20mm) and 
      +(left:20mm) .. ([yshift=-4mm]V2.west) node [left, pos=0.6] {$\Y_2$};
      
      \draw[->] ([yshift=4mm]P2.east) .. controls +(right:20mm) and 
      +(left:20mm) .. ([yshift=-4mm]V3.west) node [left, pos=0.6] {$\Y_3$};

      \draw[-] ([yshift=2mm]V3.east) .. controls +(right:20mm) .. ([yshift=2mm]M) node [above, pos=1] {output};

      \draw[-] ([yshift=-2mm]V3.east) .. controls +(right:20mm) .. ([yshift=-2mm]M) node [right, pos=1] {};

      \draw[->,densely dotted] 
      ([yshift=-4mm]P0.east) -- ([yshift=-4mm]P1.west)
      node [below, midway] {$\W_1$};
      
      \draw[->,densely dotted] 
      ([yshift=-4mm]P1.east) -- ([yshift=-4mm]P2.west)
      node [below, midway] {$\W_2$};
      
    \end{tikzpicture}
  \end{center}
  \caption{An example of a prover-verifier interaction, showing each of the objects enumerated in Section~\ref{sec:setting}. This example corresponds to the case where $r=3$. Note that Bob's channels and his intermediate memory state spaces $\W_1$ and $\W_2$ are not part of the protocol's definition. Instead, Bob is free to choose them in order to manipulate the final output bit.}
  \label{fig:interaction}
\end{figure}
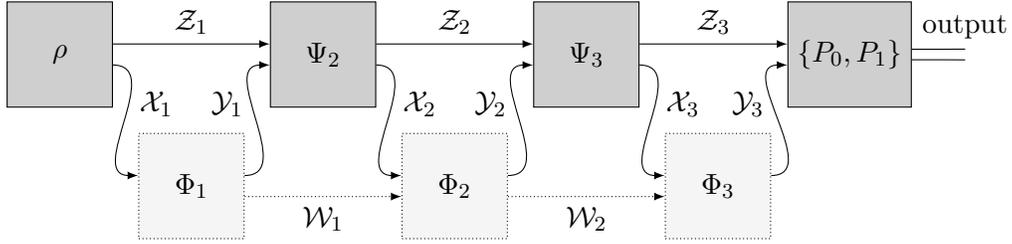

\begin{figure}
  \begin{center} \small
    \begin{tikzpicture}[scale=0.35, 
        turn/.style={draw, minimum height=14mm, minimum width=14mm,
          fill = ChannelColor, text=ChannelTextColor},
        emptyturn/.style={draw, densely dotted, minimum height=14mm, 
          minimum width=14mm, fill=black!4},
        measure/.style={draw, minimum width=7mm, minimum height=7mm,
          fill = ChannelColor},
        >=latex]

      \node (V0) at (-18,4) [turn] {$\rho$};      
      \node (V1) at (-8,4) [turn] {$\Psi_2$};
      \node (V2) at (2,4) [turn] {$\Psi_3$};
      \node (V3) at (12,4) [turn] {$\{P_0, P_1\}$};
      \node (M) at (16, 4)  {};

      \node (V02) at (-18,-6) [turn] {$\rho$};      
      \node (V12) at (-8,-6) [turn] {$\Psi_2$};
      \node (V22) at (2,-6) [turn] {$\Psi_3$};
      \node (V32) at (12,-6) [turn] {$\{P_0, P_1\}$};
      \node (M2) at (16, -6)  {};

      \node (P0) at (-13,-1) [emptyturn] {$\Phi_1$};
      \node (P1) at (-3,-1) [emptyturn] {$\Phi_2$};
      \node (P2) at (7,-1) [emptyturn] {$\Phi_3$};

      \draw[->] ([yshift=4mm]V0.east) -- ([yshift=4mm]V1.west)
      node [above, midway] {$\Z_1$};
      
      \draw[->] ([yshift=4mm]V1.east) -- ([yshift=4mm]V2.west)
      node [above, midway] {$\Z_2$};
      
      \draw[->] ([yshift=4mm]V2.east) -- ([yshift=4mm]V3.west)
      node [above, midway] {$\Z_3$};

      \draw[->] ([yshift=-4mm]V02.east) -- ([yshift=-4mm]V12.west)
      node [above, midway] {$\Z_1$};
      
      \draw[->] ([yshift=-4mm]V12.east) -- ([yshift=-4mm]V22.west)
      node [above, midway] {$\Z_2$};
      
      \draw[->] ([yshift=-4mm]V22.east) -- ([yshift=-4mm]V32.west)
      node [above, midway] {$\Z_3$};

      \draw[->] ([yshift=-4mm]V0.east) .. controls +(right:20mm) and 
      +(left:20mm) .. ([yshift=4mm]P0.west) node [right, pos=0.4] {$\X_1$};
      
      \draw[->] ([yshift=-4mm]V1.east) .. controls +(right:20mm) and 
      +(left:20mm) .. ([yshift=4mm]P1.west) node [right, pos=0.4] {$\X_2$};
      
      \draw[->] ([yshift=-4mm]V2.east) .. controls +(right:20mm) and 
      +(left:20mm) .. ([yshift=4mm]P2.west) node [right, pos=0.4] {$\X_3$};
      
      \draw[->] ([yshift=4mm]P0.east) .. controls +(right:20mm) and 
      +(left:20mm) .. ([yshift=-4mm]V1.west) node [left, pos=0.6] {$\Y_1$};
      
      \draw[->] ([yshift=4mm]P1.east) .. controls +(right:20mm) and 
      +(left:20mm) .. ([yshift=-4mm]V2.west) node [left, pos=0.6] {$\Y_2$};
      
      \draw[->] ([yshift=4mm]P2.east) .. controls +(right:20mm) and 
      +(left:20mm) .. ([yshift=-4mm]V3.west) node [left, pos=0.6] {$\Y_3$};

      \draw[-] ([yshift=2mm]V3.east) .. controls +(right:20mm) .. ([yshift=2mm]M) node [above, pos=1] {output};
      \draw[-] ([yshift=-2mm]V3.east) .. controls +(right:20mm) .. ([yshift=-2mm]M) node [right, pos=1] {};

      \draw[-] ([yshift=2mm]V32.east) .. controls +(right:20mm) .. ([yshift=2mm]M2) node [above, pos=1] {output};
      \draw[-] ([yshift=-2mm]V32.east) .. controls +(right:20mm) .. ([yshift=-2mm]M2) node [right, pos=1] {};

%--

% Arrows in 2nd instance

      \draw[->] ([yshift=+4mm]V02.east) .. controls +(right:20mm) and 
      +(left:20mm) .. ([yshift=-4mm]P0.west) node [right, pos=0.4] {$\X_1$};
      
      \draw[->] ([yshift=+4mm]V12.east) .. controls +(right:20mm) and 
      +(left:20mm) .. ([yshift=-4mm]P1.west) node [right, pos=0.4] {$\X_2$};
      
      \draw[->] ([yshift=+4mm]V22.east) .. controls +(right:20mm) and 
      +(left:20mm) .. ([yshift=-4mm]P2.west) node [right, pos=0.4] {$\X_3$};
      
      \draw[->] ([yshift=-4mm]P0.east) .. controls +(right:20mm) and 
      +(left:20mm) .. ([yshift=4mm]V12.west) node [left, pos=0.6] {$\Y_1$};
      
      \draw[->] ([yshift=-4mm]P1.east) .. controls +(right:20mm) and 
      +(left:20mm) .. ([yshift=4mm]V22.west) node [left, pos=0.6] {$\Y_2$};
      
      \draw[->] ([yshift=-4mm]P2.east) .. controls +(right:20mm) and 
      +(left:20mm) .. ([yshift=4mm]V32.west) node [left, pos=0.6] {$\Y_3$};

% --
      
      \draw[->,densely dotted] 
      ([yshift=0mm]P0.east) -- ([yshift=0mm]P1.west)
      node [below, midway] {$\W_1$};
      
      \draw[->,densely dotted] 
      ([yshift=0mm]P1.east) -- ([yshift=0mm]P2.west)
      node [below, midway] {$\W_2$};
      
    \end{tikzpicture}
  \end{center}
  \caption{Two parallel instances of the prover-verifier interaction from Figure~\ref{fig:interaction}.. We can see that Bob is allowed to arbitrarily entangle his actions between both instances, while Alice acts identically and independently.}
  \label{fig:parallelInteractions}
\end{figure}
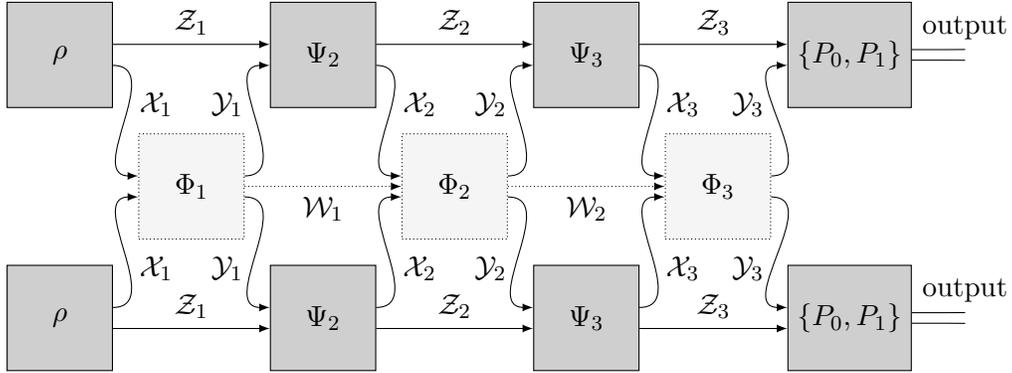

In a parallel repetition context, we have $n$ parallel instances of such a protocol. The verifier Alice will be acting identically and independently between these instances, while the prover Bob can entangle his answers between the $n$ instances. This is depicted for the case of two parallel instances in Figure~\ref{fig:parallelInteractions}.

As discussed earlier, the work in \cite{gutoski2010quantum} determines that for Bob, correlating his answers does not help when he intends to win $n$ out of $n$ parallel instances. More formally, we have the following:
\begin{theorem} [\cite{gutoski2010quantum}, Theorem~4.9]
\label{thm:theoremFromGusThesis}
Let the optimal chance for Bob of achieving the winning outcome in a prover-verifier interaction be equal to $p$.  Consider $n$ parallel instances. Then, Bob's optimal chance of achieving the winning outcome in all instances is equal to~$p^n$.
\end{theorem}
\noindent An alternative proof of this result is provided in Chapter~4 of \cite{vidick2016quantum}, and it can also be seen as a specific case of the product properties discussed in \cite{mittal2007product, lee2008product}. 

We now formally discuss interactive proof systems. In this setting, we are concerned with prover-verifier interactions where the actions of the verifier are parametrized by a string $s$, which belongs to one of two sets $L_{yes}$ and $L_{no}$. For each value of $s$, there will be an associated optimal probability $p(s)$ of Bob obtaining the winning outcome. For some constants $0 \leq a < 1$ and $a < b \leq 1$, it is a property of the proof system that if $s \in L_{no}$, then $p(s) \leq a$, while if $s \in L_{yes}$, then $p(s) \geq b$. The soundness and completeness errors are then defined as $a$ and $1-b$, respectively.  

Given such a proof system, one can seek to find another one where the soundness and completeness errors are reduced to $\epsilon$. That is to say, the new proof system corresponds to the same tuple $(L_{yes}, L_{no}$), and we have that $a \leq \epsilon$ and $b \geq 1 - \epsilon$. We will discuss in Section~\ref{sec:results} how to make usage of parallel repetition with a threshold in order to perform this task. Our main result will follow from a threshold theorem of Impagliazzo and Kabanets, together with Theorem~\ref{thm:theoremFromGusThesis}. In particular, the threshold theorem that we use is as follows:

\begin{theorem} [\cite{impagliazzo2010constructive}, Theorem 1] Let $X_1$, \ldots, $X_n$ be Boolean random variables. Assume that for some $0 \leq \delta \leq 1$ and any subset $S \subseteq \{1, \ldots, n\}$, it holds that  
\label{thm:ik2010}
\begin{equation}
\Pr\( \bigwedge_{s \in S} X_s = 1\) \leq \delta^{\abs{S}}. 
\end{equation}
Then, for any value of $\gamma$ such that $\delta \leq \gamma \leq 1$, it holds that 
\begin{equation}
\Pr\(\sum_{i=1}^n X_i  \geq \gamma n\) \leq e^{-n D(\gamma||\delta)} \leq e^{-2n (\gamma - \delta)^2} .
\end{equation}
where  $D(\gamma||\delta)$ denotes the binary relative entropy.
\end{theorem}

\section{Results}

\label{sec:results}

Combining Theorem~\ref{thm:theoremFromGusThesis} with Theorem~\ref{thm:ik2010}, we derive the following bound to quantum hedging:

\begin{theorem}
\label{thm:boundToHedging}
Consider $n$ parallel instances of an arbitrary quantum prover-verifier interaction, as defined in Section~\ref{sec:setting}. Let the optimal chance of Bob obtaining the winning outcome be equal to the constant $p$. Consider any value $k \in \mathbb{N}$ such that $k = n(p + \delta)$, with $\delta > 0$ a positive constant. Then, Bob's optimal chance of winning at least $k$ out of the $n$ parallel instances is upper-bounded by $e^{-2n\delta^2}$.
\end{theorem}

\begin{proof}
Let $X_1, \ldots, X_n$ be Boolean random variables associated with the respective outcomes of the $n$ parallel instances. It must now hold that for any subset $S \subseteq \{1, \ldots, n\}$,
\begin{equation}
\Pr\( \bigwedge_{s \in S} X_s = 1\) \leq p^{\abs{S}}.  \label{eq:boundingSum}
\end{equation}

\noindent The reason why this must be the case is because of Theorem \ref{thm:theoremFromGusThesis}. In particular, let us imagine $\abs{S}$ parallel instances of the protocol under consideration. Then, Bob can obtain a probability of winning all of these $\abs{S}$ instances equal to $\Pr\( \bigwedge_{s \in S} X_s = 1\)$. He can do that by playing as if there were $n$ parallel instances, while simulating the actions of Alice for those of the $n$ instances that do not correspond to elements of~$S$. Theorem~\ref{thm:theoremFromGusThesis} then establishes that $\Pr\( \bigwedge_{s \in S} X_s = 1\)$ must be upper-bounded by $p^{\abs{S}}$.

By Theorem \ref{thm:ik2010}, it is a consequence of Equation \eqref{eq:boundingSum} that 
\begin{equation}
\Pr\(\sum_{i=1}^n X_i  \geq k\) \leq e^{-2n (p+\delta - p)^2}  = e^{-2n\delta^2}.  \label{eq:directConsequenceIK}
\end{equation}
\end{proof}

Theorem~\ref{thm:boundToHedging} then gives us a simple proof of correctness for error reduction via parallel repetition with a threshold. In particular, let us consider an arbitrary quantum prover-verifier proof system and its corresponding values of $a$ and $b$. Then, we consider $n$ parallel instances and ask the prover to win in at least $n\frac{a+b}{2}$ instances (or equivalently, in at least $\lceil n \frac{a+b}{2} \rceil$ instances).

In order to lower the completeness and soundness errors of this procedure down to $\epsilon$, it will be enough to have $O(\log(1/\epsilon))$ parallel instances. For the completeness error, this follows from standard Chernoff bounds. For the soundness error, this follows from applying Theorem \ref{thm:boundToHedging}, with $p=a$ and $\delta=\frac{b-a}{2}$.

As previously discussed, this is of special interest for the case of protocols with two messages, where there is no known ability to switch to an equivalent protocol with the same or fewer messages and with perfect completeness.

One can alternatively choose to model the gap between $a$ and $b$ not as a constant but as an inverse polynomial in $n$. Then, when we invoke Theorem~\ref{thm:boundToHedging} we have $\delta$ be an inverse polynomial in $n$, and the resulting bound on the number of parallel instances is $O\((1/\epsilon) \frac{1}{\delta^2}\) = O((1/\epsilon) poly(n))$.

\section*{Acknowledgements}
This work was performed with funding from Canada’s NSERC and a University of Waterloo President’s Graduate Scholarship. Thanks are due to Eric Blais for suggesting the possible relevance of \cite{impagliazzo2010constructive}, and to John Watrous for encouraging the writing of this note and providing valuable feedback.

\bibliographystyle{alpha}
\bibliography{refs}

\end{document}